\documentclass[10pt]{amsart}
\usepackage{amsmath, amssymb}
\usepackage[shortlabels]{enumitem} % [shortlabels] supports changing enumerate labels easily
\usepackage{mathrsfs} % get script L, H
\usepackage[margin=1.15in]{geometry}
\usepackage{bbm} % to get \mathbbm{1}
\usepackage{tikz}
\usepackage{hyperref}  % KEEP THIS PACKAGE LAST

\newcommand{\comments}[1]{}
\newcommand{\N}{\mathbb{N}}

\newcommand{\R}{\mathbb{R}}
\newcommand{\C}{\mathbb{C}}

\renewcommand{\Im}{\operatorname{\mathfrak{Im}}}
\renewcommand{\Re}{\operatorname{\mathfrak{Re}}}

\renewcommand{\tilde}{\widetilde}
\renewcommand{\epsilon}{\varepsilon}

\newtheorem{thm}{Theorem}
\newtheorem{prop}{Proposition}
\newtheorem{cor}{Corollary}
\newtheorem{lemma}{Lemma}

\newtheorem*{example*}{Example}

\theoremstyle{definition}

\theoremstyle{remark}
\newtheorem*{rmk}{Remark}
\newtheorem*{rmks}{Remarks}

\numberwithin{equation}{section}
\begin{document}

\title{Trace formulas for Schr\"odinger operators on star graphs}
%\date{October 27, 2015}

\author{Semra Demirel-Frank}
\address{Semra Demirel-Frank, Mathematics 253-37, Caltech, Pasadena, CA 91125}
\email{sdemirel@caltech.edu}

\author{Laura Shou}
\address{Laura Shou, Mathematics 253-37, Caltech, Pasadena, CA 91125}
\email{lshou@caltech.edu}

\begin{abstract}
We derive trace formulas of the Buslaev-Faddeev type for quantum star graphs. One of the new ingredients is high energy asymptotics of the perturbation determinant.
\end{abstract}

\maketitle

\section{Introduction}

A \emph{quantum graph} is a triple $(\Gamma,H,\mathcal{VC})$, where $\Gamma$ is a metric graph with edges $\{e_j\}$, $H$ is a differential operator, and $\mathcal{VC}$ is a set of vertex conditions. %We'll denote the edges of $\Gamma$ by $\{e_j\}$. 
The usual differential operator is the Schr\"{o}dinger operator defined by
\begin{equation}
H\psi=-\frac{d^2}{dx^2}\psi+V\psi,
\end{equation}
where $\psi\in L^2(\Gamma):=\bigoplus_j L^2(e_j)$ and $V$ is a real-valued potential on $\Gamma$, with restriction to edge $e_j$ denoted by $v_j:=V|_{e_j}$. We assume throughout that $V$ is sufficiently regular, for instance, in $L^1_{\text{loc},\text{unif}}$ on each edge. We can require continuity of $\psi$ on $\Gamma$ and impose Kirchoff boundary conditions on each vertex $v$,
\begin{equation}\label{eqn:kirchoff}
\sum_e\frac{d\psi}{dx_{e+}}(v)=0,
\end{equation}
where the sum is over all edges $e$ containing the vertex $v$, and the $e+$ indicates the derivative is taken in the outgoing direction. %In other words, the net flux at each vertex is zero. 
These conditions make $H$ a self-adjoint operator on $L^2(\Gamma)$. % [no] statement about $V$ regular, finitely many evs

Besides being interesting mathematically, quantum graphs have many applications to chemistry and physics. Since the 1930s, they have been used to model structures ranging from nanotechnology and quantum wires to free electrons in organic molecules. Since quantum graphs have one dimensional edges, they are also used as simplified models for complicated behavior in physics, like quantum chaos and Anderson localization \cite{book}.

One of the simplest examples of a quantum graph is a quantum \emph{star graph}, which consists of a single vertex $v$ and $n\ge 2$ edges $e_j$, $1\le j\le n$, each of which is identified with the half-line $[0,\infty)$. %The continuity and Kirchoff boundary conditions then become,
%\begin{equation}
%\psi_1(0)=\cdots=\psi_n(0),\quad\sum_{j=1}^n\psi_j'(0)=0,
%\end{equation}
\comments{
\begin{figure}[h!]
\begin{center}
\begin{tikzpicture}
\draw [fill] (0,0) circle (.075) node [below right] {$v$};
\draw [ultra thick, ->] (0,0)--(1.4,0);
\draw [ultra thick, ->] (0,0)--(1,1);
\draw [ultra thick, ->] (0,0)--(-.1,1.3);
\draw [ultra thick, ->] (0,0)--(-1.2,-.8);
\draw [ultra thick, ->] (0,0)--(0,-1.4);
\end{tikzpicture}
\end{center}
\caption{A star graph.}
\end{figure}
}
Because of its similarity to $n$ half-lines, many of the spectral theory and scattering theory results for the half-line case have analogues for star graphs.
For example, Levinson's formula for star graphs and expressions for the spectral shift function and perturbation determinant for star graphs were proved in \cite{demirel-ssf}. In this article, we will prove trace formulas of the Buslaev-Faddeev type for star graphs. We will first briefly review the scattering theory that is relevant to us here; see also \cite[\S4]{yafaev}. %In this case, an incident wave $e^{-ikx}$ comes in from $+\infty$, interacts with the potential $v(x)$, and then is reflected by the boundary at $x=0$. The outgoing reflected wave is $s(k)e^{ikx}$ with $|s(k)|=1$. 
In the one dimensional case, the Schr\"{o}dinger equation is just
\[
-u''+V(x)u=zu,\quad z=\zeta^2.
\]
Consider a half-line edge $e_j$ in the star graph with
$$
\int_{e_j}|v_j(x)|\,dx<\infty.
$$
Then there are two linearly independent solutions to the Schr\"{o}dinger equation on $e_j$, the regular solution $\varphi_j$ and the \emph{Jost solution} $\theta_j$. The Jost solution is determined by the asymptotics,
\begin{align}
\theta_j(x,\zeta)&=e^{ix\zeta}(1+o(1)),\quad x\to\infty,\label{eqn:jost}\\
\theta_j'(x,\zeta)&=i\zeta e^{ix\zeta}(1+o(1)),\quad x\to\infty,\label{eqn:jost-derivative}
\end{align}
and gives rise to the \emph{Jost function},
$$
\omega_j(\zeta):=\theta(0,\zeta).
$$
In the half-line case and under suitable conditions on the potential, the (modified) perturbation determinant $D(\zeta)$ (Section~\ref{sec:known}) is simply given by $D(\zeta)=\omega(\zeta)$ and is analytic in $\zeta$ for $\Im\zeta>0$ and continuous up to $\Im\zeta=0$ except possibly at $\zeta=0$. For the star graph, the perturbation determinant is more complicated, but an explicit formula in terms of the Jost functions and solutions is given in \cite{demirel-ssf}, which we will use in Section~\ref{sec:known}.

Using the perturbation determinant, we can define the \emph{limit amplitude} $a(k)$ and \emph{limit phase} (or \emph{phase shift}) $\eta(k)$ for $\zeta=k\in\R$. We set $D(k)=:a(k)e^{i\eta(k)}$, with $a(k)=|D(k)|$. As we will see later in \eqref{eqn:D(z)-asym}, $D(z)=1+O(|\zeta|^{-1})$ as $|\zeta|\to\infty$, so we can choose $\eta$ with the convention $\eta(\infty)=0$. (We also require it to be continuous.)
The names for $a(k)$ and $\eta(k)$ come from scattering theory, where they correspond to the amplitude shift and phase shift in the wavefunction as $x\to\infty$.

The last thing we mention from scattering theory is zero energy resonances. The Schr\"{o}dinger operator $H$ has a \emph{zero energy resonance} if there is a nontrivial bounded solution to $-\frac{d^2}{dx^2}\psi+V\psi=0$ that satisfies the continuity and Kirchhoff conditions. The \emph{multiplicity of the resonance} is the dimension of the solution space. In physics, these resonances correspond to ``half-bound states'' or ``metastable bound states''.

Our main result is the following trace formulas of the type in \cite{BuslaevFaddeev} relating $\sum|\lambda_j|^{n}$ (for $n\in\frac{1}{2}\N$) to an expression involving the potential $V$.
\begin{thm}[trace formulas]\label{thm:trace}
Let $\Gamma$ be a star graph with edges $\{e_j\}$, $1\le j\le n$.
Assume that 
\begin{equation}\label{eqn:potential}
\int_{e_j}(1+x)|v_j(x)|\,dx<\infty,\quad 1\le j\le n,
\end{equation}
and that for each $j=1,\ldots,n$ and $m\in\N_0:=\N\cup\{0\}$,
% this is redundant if we assume 
\begin{equation}\label{eqn:potential-smooth}
|v_j^{(m)}(x)|\le C_m(1+|x|)^{-\rho-m},\quad \text{some }\rho\in(1,2].
\end{equation}
%\begin{equation}\label{eqn:potential}
%\int_{e_j}|v_j(x_j)|\,dx_j<\infty,\quad 1\le j\le n.
%\end{equation}
If $\zeta=0$ is a resonance of multiplicity one, we also assume that
\begin{equation}\label{eqn:potential2}
\int_{e_j}(1+x^2)|v_j(x)|\,dx<\infty,\quad 1\le j\le n.
\end{equation} 
Then letting $r_j$ be the multiplicity of the eigenvalue $\lambda_j$, we have
\begin{equation}
\sum_{j=1}^Nr_j|\lambda_j|^{1/2}-\pi^{-1}\int_0^\infty\log a(k)\,dk=\frac{L_1}{4}\equiv-\frac{1}{4}\sum_{j=1}^n\int_0^\infty v_j(x)\,dx.
\end{equation}
For $m=1,2,\ldots$,
\begin{equation}
\sum_{j=1}^Nr_j|\lambda_j|^m+(-1)^m\pi^{-1}2m\int_0^\infty\left(\eta(k)-\sum_{j=0}^{m-1}(-1)^{j+1}L_{2j+1}(2k)^{-2j-1}\right)k^{2m-1}\,dk=-m2^{-2m}L_{2m}
\end{equation}
\begin{multline}
\sum_{j=1}^Nr_j|\lambda_j|^{m+1/2}+(-1)^{m+1}\pi^{-1}(2m+1)\int_0^\infty\left(\log a(k)-\sum_{j=1}^m(-1)^jL_{2j}(2k)^{-2j}\right)k^{2m}\,dk\\
=(2m+1)2^{-2m-2}L_{2m+1}.
\end{multline}
The coefficients $L_m$ come from the asymptotic expansion $\log D(\zeta)=\sum_{m=1}^\infty L_m(2i\zeta)^{-m}$, for $|\zeta|\to\infty$. The first few coefficients are given by
\begin{multline}\label{eqn:Lm}
L_1=-\sum_{j=1}^n\int_0^\infty v_j(x)\,dx,\quad L_2=\sum_{j=1}^nv_j(0)\left[\frac{2}{n}-1\right],\quad
L_3=\sum_{j=1}^n\left(v_j'(0)\left[1-\frac{2}{n}\right]+\int_0^\infty v_j^2(x)\,dx\right),\hfill{}\\
L_4=\sum_{j=1}^n\left(v_j''(0)\left[\frac{2}{n}-1\right]-v_j^2(0)\left[\frac{2}{n}-2\right]\right)-\frac{2}{n^2}\left(\sum_{j=1}^n v_j(0)\right)^2,\hfill{}\\
L_5=\sum_{j=1}^n\left(v_j'''(0)\left[1-\frac{2}{n}\right]+v_j(0)v_j'(0)\left[\frac{8}{n}-6\right]\right)+\frac{4}{n^2}\left(\sum_{j=1}^nv_j(0)\sum_{j=1}^nv_j'(0)\right)-\sum_{j=1}^n\int_0^\infty(v_j'(x)^2+2v_j^3(x))\,dx.
\end{multline} % [resolved] $L_5$ computation checked 09.28.15

\end{thm}
\begin{rmks}\mbox{}
\begin{enumerate}[(i)]
\item These trace formulas are analogous to the trace formulas for the half-line case, which can be found in \cite[\S4.6]{yafaev}. The differences are the values of the $L_m$'s and the possibility of having eigenvalues of multiplicity greater than one. %Yafaev pg.171 Lemma 2.2 half-line case they are simple zeros
%The key steps to prove the theorem are to use an expression for the perturbation determinant $D(\zeta)$ from \cite{demirel-ssf} to find the high energy asymptotic expansion of $D(\zeta)$, and then adapt the results from \cite[\S4.6]{yafaev} for the star graph.
\item The additional requirement \eqref{eqn:potential2} in the case that $\zeta=0$ is a resonance of multiplicity one comes from the hypotheses in \cite[Prop.4.5]{demirel-ssf} for low energy asymptotics of the perturbation determinant.
\end{enumerate}
\end{rmks}

We emphasize that while we are requiring each $v_j$ to be smooth on $e_j$, we do not impose restrictions on the values of $v_j$ and its derivatives at the vertex. This raises the question of whether the vertex terms in the coefficients $L_m$ disappear if $V$ is smooth at the vertex. A natural notion of smoothness on $\Gamma$ is that for any two distinct $e_i$ and $e_j$, the function on $\R$ obtained by combining $v_i$ and $v_j$ is smooth. This is easily seen to be equivalent to $v_i^{(2k)}(0)=v_j^{(2k)}(0)$ for all $1\le i,j\le n$ and $v_i^{(2k+1)}(0)=0$ for all $1\le i\le n$. We shall see that while for the first few coefficients $L_m$ with $m$ odd the vertex terms cancel, they do not for $m$ even.

\begin{cor}
For $k\in\N_0:=\N\cup\{0\}$, consider the special case where $v_i^{(2k)}(0)=v_j^{(2k)}(0)=:v^{(2k)}(0)$ for all $1\le i,j\le n$ and $v_i^{(2k+1)}(0)=0$ for all $1\le i\le n$. Then
\begin{multline}
L_1=-\sum_{j=1}^n\int_0^\infty v_j(x)\,dx,\quad L_2=v(0)\left[2-n\right],\quad L_3=\sum_{j=1}^n\int_0^\infty v_j^2(x)\,dx,\hfill{}\\ L_4=v''(0)\left[2-n\right]-v^2(0)\left[4-2n\right],\quad L_5=-\sum_{j=1}^n\int_0^\infty(v_j'(x)^2+2v_j^3(x))\,dx.\hfill{}
\end{multline}
\end{cor}

Let us comment the outline of this paper and on our approach to Theorem \ref{thm:trace}. There is a well-known strategy to obtain trace formulas which goes back to \cite{BuslaevFaddeev} (see, for instance, \cite{yafaev} for a textbook presentation) and which we will follow here. It consists in the following steps:
\begin{enumerate}
\item One integrates the perturbation determinant over a contour and takes limits in the form of the contour to obtain a family of identities.
\item One analytically continues these identities and evaluates them at certain points.
\end{enumerate}
A key step both in Steps 1 and 2 is to find a representation of the perturbation formula in terms of Jost solutions. This allows to prove bounds on the perturbation determinant which are necessary both to control the limit of the contour in Step 1 and to perform the analytic continuation in Step 2. In order to carry out this program we will rely on \cite{demirel-ssf}, which contains a useful formula for the perturbation determinant on a star graph, see Proposition \ref{prop:D(z)}, and gives the low energy asymptotics of the perturbation determinant, see Proposition \ref{prop:low-energy}. Using these results we can carry out Step 1 in a rather straightforward manner. For Step 2, however, we also need high energy asymptotics for the perturbation determinant, and that is our main technical result in this paper. While they also rely on the representation formula from \cite{demirel-ssf} they require an inductive procedure and careful remainder estimates. We provide those in Section \ref{sec:high-energy} and in Appendix \ref{appendix:remainders}.

In order to make this paper self-contained we review necessary results from the literature in Section~\ref{sec:known} and provide details for Steps 1 and 2 in Section~\ref{sec:trace}. As we already mentioned, Section \ref{sec:high-energy} contains the novel high energy asymptotics which lead, in particular, to the formulas for the coefficients $L_m$ from \eqref{eqn:Lm}. In Section~\ref{sec:real}, we use the star graph with $n=2$ to recover some of the results for the whole real line that are given in \cite[\S5]{yafaev}. In contrast to \cite[\S5]{yafaev} we also obtain formulas if $V$ is smooth away from a point, and we see explicitly the contribution to the trace formulas of the discontinuities of $V$ and its derivatives at this point.

%%%%%%%%%%%%%%%%%%%%%%%%%%%%%%%%%%%%%%%%%%%%%
%%%%%%%%%%%%%%%%%%%%%%%%%%%%%%%%%%%%%%%%%%%%%
\section{The perturbation determinant and other previous results}\label{sec:known}
%%%%%%%%%%%%%%%%%%%%%%%%%%%%%%%%%%%%%%%%%%%%%
%%%%%%%%%%%%%%%%%%%%%%%%%%%%%%%%%%%%%%%%%%%%%

The \emph{free} Schr\"{o}dinger operator is just the operator $H_0:=-\frac{d}{dx^2}$, i.e. there is no potential $V$. The corresponding resolvent will be denoted by $R_0(z):=(H_0-z)^{-1}$, while the resolvent for $H$ will be denoted by $R(z):=(H-z)^{-1}$. The \emph{perturbation determinant} (PD), introduced by Krein in 1953,
produces a holomorphic function on the resolvent set $\rho(H_0)$ that is determined by the pair of operators $H,H_0$.  In our case, we will actually need to look at the \emph{modified} perturbation determinant (see \cite{yafaev1,yafaev} for details). It is closely related to the spectral shift function, which has applications to many areas, including spectral theory, scattering theory, and trace formulas \cite{app-ssf}. The modified perturbation determinant, which we will just call the perturbation determinant, is given by
\[
D(\zeta):=\det(\mathbbm{1}+\sqrt{V}R_0(\zeta^2)\sqrt{|V|}),\quad\Im\zeta>0,
\]
where $\sqrt{V}:=(\operatorname{sgn} V)\sqrt{|V|}$. %Besides being related to the spectral shift function, it also satisifes a nice trace equation, 
%\[
%\frac{D'(z)}{D(z)}=\operatorname{tr}(R_0(z)-R(z)),\quad z\in\rho(H_0)\cap\rho(H).
%\]
It is shown in \cite{demirel-ssf} that this is well-defined in our case.
Some useful facts about the perturbation determinant that can be found in e.g. \cite{yafaev1,yafaev} are as follows.
\begin{itemize}
\item $D(\zeta)$ is holomorphic in $\Im\zeta>0$.
%\item $D(\zeta)$ has boundary values as $\zeta\to\lambda+i0$ for a.e. $\lambda\in\R$. (This uses \cite[Thm 3.2]{demirel-ssf}.) TODO check this z vs \zeta
\item $D(\zeta)$ has a zero in $\zeta$ of order $r$ if and only if $\zeta^2$ is an eigenvalue of multiplicity $r$ of $H$.
\item $D^{-1}(\zeta)\frac{d}{d\zeta^2}D(\zeta)=\operatorname{tr}(R_0(\zeta^2)-R(\zeta^2))$.
\end{itemize}
\begin{rmk}
The perturbation determinant is sometimes defined with an argument of $z=\zeta^2$. Then the definition is $D(z):=\det(\mathbbm{1}+\sqrt{V}R_0(z)\sqrt{|V|})$ for $z\in\rho(H_0)$, and the last property listed above takes on a nicer form. We will however continue to use the definition with $\zeta$.
\end{rmk}

\begin{prop}[formula for the PD, \cite{demirel-ssf}]\label{prop:D(z)}
If the potential $V$ is a real-valued function on $\Gamma$ such that
\eqref{eqn:potential} holds, then for $\zeta^2=z\in\rho(H_0)$,
\begin{equation}\label{eqn:D(z)}
D(\zeta)=\frac{K(\zeta)}{in\zeta}\prod_{j=1}^n\omega_j(\zeta),\quad\Im\zeta>0,
\end{equation}
where $K(\zeta):=\sum_{j=1}^n\frac{\theta_j'(0,\zeta)}{\theta_j(0,\zeta)}$.
\end{prop}

The importance of this proposition is that it connects a spectral theoretic object, namely $D(\zeta)$, with ODE objects, namely the $\theta_j$'s. Results of this type go back to \cite{JoPa}; see also \cite{GeMiZi}.

\begin{rmk}
Although $D(\zeta)$ is not defined for $\Im\zeta=0$, we can use \eqref{eqn:D(z)} to extend it to $\Im\zeta=0$, $\zeta\ne0$.
Then for $\zeta=k\in\R\setminus\{0\}$, define $D(k)=:a(k)e^{i\eta(k)}$, with $a(k):=|D(k)|$. Either using general properties of $D(\zeta)$ or \eqref{eqn:D(z)} and the fact  $\theta_j(x,-k)=\overline{\theta_j(x,k)}$, we get,
\begin{equation}\label{eqn:even-odd}
D(-k)=\overline{D(k)},\quad \eta(-k)=-\eta(k),\quad a(-k)=a(k).
\end{equation}
%which will be useful in the next section.
\end{rmk}

\begin{prop}[low energy asymptotics, \cite{demirel-ssf}]\label{prop:low-energy}
Assume the potential $V$ satisfies \eqref{eqn:potential}. Let $m$ be the multiplicity of the resonance $\zeta=0$, with the convention that $m=0$ if $\zeta=0$ is not a resonance. If $\zeta=0$ is a resonance of multiplicity one, we also assume that $V$ has a second moment, i.e. \eqref{eqn:potential2} holds.
Then as $\zeta\to0$, 
\begin{equation}
D(\zeta)=c\zeta^{m-1}(1+o(1)),
\end{equation}
with $c\ne0$.
\end{prop}

Proposition \ref{prop:D(z)} also yields the leading order of the high energy asymptotics of $D(\zeta)$. In fact, from asymptotics for the half-line, $\omega_j(\zeta)=\theta_j(0,\zeta)=1+O(|\zeta|^{-1})$ and $\theta'(0,\zeta)=i\zeta+O(1)$ as $|\zeta|\to\infty$, we have
\begin{equation}
K(\zeta)=ni\zeta+O(1),\quad|\zeta|\to\infty,\quad\Im\zeta\ge0,
\end{equation}
which implies
\begin{equation}\label{eqn:D(z)-asym}
D(\zeta)=\frac{ni\zeta+O(1)}{in\zeta}\prod_{j=1}^n(1+O(|\zeta|^{-1}))=1+O(|\zeta|^{-1}),
\end{equation}
which is the same limiting behavior as in the half-line case.
To get the coefficients $L_m$ in Theorem~\ref{thm:trace}, we need a full asymptotic expansion for $K(\zeta)$, which we will compute in Section~\ref{sec:high-energy}.

%%%%%%%%%%%%%%%%%%%%%%%%%%%%%%%%%%%%%%%%%%%%%
%%%%%%%%%%%%%%%%%%%%%%%%%%%%%%%%%%%%%%%%%%%%%
\section{Adapting results from the half-line case}\label{sec:trace}
%%%%%%%%%%%%%%%%%%%%%%%%%%%%%%%%%%%%%%%%%%%%%
%%%%%%%%%%%%%%%%%%%%%%%%%%%%%%%%%%%%%%%%%%%%%

Trace formula derivations for the half-line case can be found in \cite[\S4.6]{yafaev}. We will follow the same method here, but with some adaptations to ensure the results hold for star graphs. We will assume we know the coefficients $L_m$ in the asymptotic expansion of $\log D(\zeta)$, which will result in proving Theorem~\ref{thm:trace} except for the formulas \eqref{eqn:Lm} for the $L_m$'s. (The expressions for the $L_m$'s will be derived in Section~\ref{sec:high-energy}.)
\begin{lemma}
Assume \eqref{eqn:potential} holds, and also \eqref{eqn:potential2} holds if $m=1$. For $s\in\C$, $0<\Re s<\frac{1}{2}$, set
\[
F(s):=\int_0^\infty\log a(k)k^{2s-1}\,dk,\quad G(s):=\int_0^\infty\eta(k)k^{2s-1}\,dk.
\]
Then
\begin{equation}\label{eqn:FG}
F(s)\sin(\pi s)-G(s)\cos(\pi s)=\frac{\pi}{2s}\sum_{j=1}^Nr_j|\lambda_j|^s,
\end{equation}
where $r_j$ is the multiplicity of the eigenvalue $\lambda_j$.
\end{lemma}
\begin{proof}
The proof is essentially the same as in the half-line case, by applying the residue theorem to $\int_{\Gamma_{R,\varepsilon}}\frac{\frac{d}{d\zeta}D(\zeta)}{D(\zeta)}\zeta^{2s}\,dq$. As in the half-line case, the contour $\Gamma_{R,\varepsilon}$ consists of the semi-circles of radius $R$ and $\varepsilon<R$ in the upper half-plane, connected by $[-R,-\varepsilon]\cup[\varepsilon,R]$. (Figure~\ref{fig:contour}.) There are two main differences for a star graph: $D$ may have zeros that are not simple, and integration over the contour $C_\varepsilon$ must use different low energy asymptotics for $D(\zeta)$. 

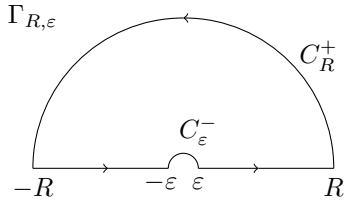
\begin{figure}[!ht]
\begin{center}
\begin{tikzpicture}
\draw [->] (2,0) arc (0:90:2);
\draw (0,2) arc (90:180:2);
\draw [->] (-2,0)--(-1,0);
\draw (-1,0)--(-.2,0);
\draw (-.2,0) arc (180:0:.2); 
\node [below] at (.2,0) {$\varepsilon$};
\node [below] at (-.3,.05) {$-\varepsilon$};
\draw [->] (.2,0)--(1,0);
\draw (1,0)--(2,0) node [below] {$R$};
\node [below] at (-2,0) {$-R$};
\node at (1.8,1.5) {$C_R^+$};
\node at (.2,.5) {$C_\varepsilon^-$};
\node at (-2,2) {$\Gamma_{R,\varepsilon}$};
%
%\draw (0,.2)--(0,2.5);
\end{tikzpicture}
\end{center}
\caption{The contour $\Gamma_{R,\varepsilon}$.}\label{fig:contour}
\end{figure}

For $R,\varepsilon$ chosen so the contour contains all the zeros of $D(\zeta)$, the residue theorem yields,
\[
\int_{\Gamma_{R,\varepsilon}}\frac{\frac{d}{d\zeta}D(\zeta)}{D(\zeta)}\zeta^{2s}d\zeta=2\pi i\sum_{j=1}^N\operatorname{Res}_{\zeta=i\kappa_j}\frac{\frac{d}{d\zeta}D(\zeta)}{D(\zeta)}\zeta^{2s}=2\pi i\sum_{j=1}^Nr_je^{\pi is}\kappa_j^{2s},
\]
where $i\kappa_j=i|\lambda_j|^{1/2}$ is a zero of $D(z)$ of order $r_j$.
%Let $i\kappa_j$ be a zero of $D(z)$ of order $r_j$; then we can compute,
%\begin{equation}\label{eqn:residue}
%\operatorname{Res}_{\zeta=i\kappa_j}\frac{\frac{d}{d\zeta}D(\zeta)}{D(\zeta)}\zeta^{2s}=r_j(i\kappa_j)^{2s}=r_je^{i\pi s}{\kappa_j^{2s}}.
%\end{equation}
%This can be shown explicitly by e.g. long division of power series, since if $D(\zeta)=(\zeta-i\kappa_j)^{r_j}(a_r+a_{r+1}(\zeta-i\kappa_j)+\cdots)$ where $a_r\ne0$, then
%\[
%\frac{\frac{d}{d\zeta}D(\zeta)}{D(\zeta)}=\frac{1}{\zeta-i\kappa_j}\left[r_j+\frac{a_{r+1}}{a_r}(\zeta-i\kappa_j)+\cdots\right].
%\]
%From \eqref{eqn:residue}, we get
%\begin{equation}
%\int_{\Gamma_{R,\varepsilon}}\frac{\frac{d}{d\zeta}D(\zeta)}{D(\zeta)}\zeta^{2s}\,d\zeta=2\pi i\sum_{j=1}^Nr_je^{\pi is}\kappa_j^{2s}.
%\end{equation}
Integration by parts along a semicircle $C_{r}^+$ of radius $r$ yields, % from Yafaev book
\[
\int_{C_r^+}\frac{\frac{d}{d\zeta}D(\zeta)}{D(\zeta)}\zeta^{2s}\,d\zeta=-2s\int_{C_r^+}\zeta^{2s-1}\log D(\zeta)\,d\zeta+(-r)^{2s}\log D(-r)-r^{2s}\log D(r).
\]
If $r=R$, this integral goes to zero for $\Re s<\frac{1}{2}$ since $D(\zeta)=1+O(|\zeta|^{-1})$; this is the same as in the half-line case. For $r=\varepsilon$, %we need to use low energy asymptotics for $D(\zeta)$, which are different than in the half-line case. 
recall from Proposition~\ref{prop:low-energy} that $D(\zeta)=c\zeta^{m-1}(1+o(1))$ as $\zeta\to0$, and $c\ne0$, so
\[
|\varepsilon^{2s}\log D(\varepsilon)|=|\log(c\varepsilon^{m-1}(1+o(1)))|\varepsilon^{2\Re s}\le(|\log\varepsilon^{m-1}|+|\log|c||+o(1)+\pi)\varepsilon^{2\Re s}\to0
\]
as $\varepsilon\to0$, $m=0,1,2,\ldots$, since $\Re s>0$. Similarly,
\[
\left|\int_{C_\varepsilon^+}\zeta^{2s-1}\log D(\zeta)\,d\zeta\right|\lesssim|\log D(\varepsilon)|\varepsilon^{2\Re s-1}\cdot\varepsilon\to0.
\]
%We will show that this integral goes to zero as $r=R\to\infty$ or $r=\varepsilon\to0$.
%\begin{enumerate}[(i)]
%\item $r=R\to\infty$: Recall \eqref{eqn:D(z)-asym} that $D(\zeta)=1+O(|\zeta|^{-1})$ as $|\zeta|\to\infty$, which is the same estimate as in the half-line case. So if $\Re s<\frac{1}{2}$, then $|(\log D(\pm R))(\pm R)^{2s}|\sim|\log(1+O(|\zeta|^{-1}))| R^{2\Re s}\sim\frac{1}{R}R^{2\Re s}\to0$ as $R\to\infty$. Similarly, we also have
%\[
%\left|\int_{C_R^+}\log D(\zeta)\zeta^{2s-1}\,d\zeta\right|\lesssim|\log D(\zeta)|R^{2\Re s-1}\cdot R\sim\log(1+O(|\zeta|^{-1}))R^{2\Re s}\to0.
%\]
%
%
%\item For $r=\varepsilon$, we need to use low energy asymptotics for $D(\zeta)$, which are different than in the half-line case. The end result here is the same though. Recall from Proposition~\ref{prop:low-energy} that $D(\zeta)=c\zeta^{m-1}(1+o(1))$ as $\zeta\to0$, and $c\ne0$. Then,
%\[
%|\log D(\varepsilon)\varepsilon^{2s}|\sim|\log(c\varepsilon^{m-1}(1+o(1)))|\varepsilon^{2\Re s}\le(|\log\varepsilon^{m-1}|+|\log|c||+o(1)+\pi)\varepsilon^{2\Re s}\to0.
%\]
%as $\varepsilon\to0$, $m=0,1,2,\ldots$, since $\Re s>0$. Similarly,
%\[
%\left|\int_{C_\varepsilon^+}\log D(\zeta)\zeta^{2s-1}\,d\zeta\right|\lesssim|\log D(\varepsilon)|\varepsilon^{2\Re s-1}\cdot\varepsilon\to0.
%\]
%\end{enumerate}
Now by taking $R\to\infty$ and $\varepsilon\to0$, we obtain
\[
\int_{-\infty}^\infty\frac{\frac{d}{d\zeta}D(\zeta)}{D(\zeta)}\zeta^{2s}\,d\zeta=2\pi i\sum_{j=1}^Nr_je^{\pi is}\kappa_j^{2s}.
\]
Integrating by parts and using $\log D(\zeta)=\log a(k)+i\eta(k)$ along with \eqref{eqn:even-odd} yields the desired formula.

\end{proof}

We can analytically continue $F$ and $G$ to $\Re s>0$ just as in the half-line case. Suppose we have the asymptotic expansion
\begin{equation*}
\log D(\zeta)=\sum_{m=1}^\infty L_m(2i\zeta)^{-m},\quad|\zeta|\to\infty,\;\Im\zeta\ge0.
\end{equation*}
%(We will compute expressions for the $L_m$'s in Section~\ref{sec:high-energy}.) 
Then
\begin{equation*}
\log a(k)=\sum_{j=1}^\infty(-1)^jL_{2m}(2k)^{-2j},\quad \eta(k)=\sum_{j=0}^\infty(-1)^{j+1}L_{2m+1}(2k)^{-2j-1},
\end{equation*}
%This is like the half-line case, except of course the $L_m$'s are different. 
%Assuming we have this asymptotic expansion, we get the same result as in \cite[\S4-Lemma 6.4]{yafaev}:
and we can get the same result for graphs as in \cite[\S4-Lemma 6.4]{yafaev}:
\begin{lemma}[analytic continuation]\label{lem:analytic-continuation}
Assume the potential satisfies \eqref{eqn:potential} (integrability) and \eqref{eqn:potential-smooth}. If $m=1$, also assume it satisfies \eqref{eqn:potential2} (second moment). Then the functions $F(s)$ and $G(s)$ are meromorphic in the half-plane $\Re s>0$. The function $F(s)$ is analytic everywhere except integers $m=1,2,\ldots$, where it has simple poles with residues $(-1)^{m+1}2^{-2m-1}L_{2m}$. If $\Re s<1$, then the original definition is true. Otherwise, if $m<\Re s<m+1$, then
\begin{equation}
F(s)=\int_0^\infty(\log a(k)-\sum_{j=1}^m(-1)^jL_{2j}(2k)^{-2j})k^{2s-1}\,dk.
\end{equation}
The function $G(s)$ is analytic everywhere except half-integer points $m+\frac{1}{2}$, $m=0,1,2,\ldots$, where it has simple poles with residues $(-1)^m2^{-2m-2}L_{2m+1}$. If $m\ge1$ and $m-\frac{1}{2}<\Re s<m+\frac{1}{2}$, then
\begin{equation}
G(s)=\int_0^\infty(\eta(k)-\sum_{j=0}^{m-1}(-1)^{j+1}L_{2j+1}(2k)^{-2j-1})k^{2s-1}\,dk.
\end{equation}
\end{lemma}

\begin{proof}[Proof (of most of Theorem~\ref{thm:trace})]
By analytic continuation, \eqref{eqn:FG} holds for all $\Re s>0$. Evaluating $s=m=1,2,\ldots$ and $s=m+1/2$ for $m=0,1,2,\ldots$, yields
\begin{align*}
\pi\operatorname{Res}_{s=m}F(s)-G(m)&=(-1)^m\frac{\pi}{2m}\sum_{j=1}^Nr_j|\lambda_j|^m\\
F(m+1/2)+\pi\operatorname{Res}_{s=m+1/2}G(s)&=(-1)^m\frac{\pi}{2m+1}\sum_{j=1}^N|\lambda_j|^{m+1/2}.
\end{align*}
Using Lemma~\ref{lem:analytic-continuation} yields the trace formulas in Theorem~\ref{thm:trace}.
\end{proof}

%%%%%%%%%%%%%%%%%%%%%%%%%%%%%%%%%%%%%%%%%%%%%
%%%%%%%%%%%%%%%%%%%%%%%%%%%%%%%%%%%%%%%%%%%%%
\section{High energy asymptotic expansions}\label{sec:high-energy}
%%%%%%%%%%%%%%%%%%%%%%%%%%%%%%%%%%%%%%%%%%%%%
%%%%%%%%%%%%%%%%%%%%%%%%%%%%%%%%%%%%%%%%%%%%%
In this section, we will compute the coefficients in the asymptotic series expansion,
\begin{equation}\label{eqn:log-asym}
\log D(\zeta)=\sum_{m=1}^\infty L_m(2i\zeta)^{-m},\quad|\zeta|\to\infty,\;\Im\zeta\ge0.
\end{equation}
To this end, assume \eqref{eqn:potential-smooth} holds for each $j=1,\ldots,n$ and $m\in\N_0$.
Using the formula for the PD \eqref{eqn:D(z)}, we have
\begin{equation}\label{eqn:log-sum}
\log D(\zeta)=\log\left(\frac{K(\zeta)}{in\zeta}\right)+\sum_{j=1}^n\log(\omega_j(\zeta)),
\end{equation}
which is permissible since each separated term is $1+O(|\zeta|^{-1})$ and hence has argument close to zero as $|\zeta|\to\infty$.
From \cite[\S4]{yafaev}, there is the following asymptotic series for the logarithm of each Jost function $\omega_j(\zeta)$,
\[
\log(\omega_j(\zeta))=\sum_{m=1}^\infty\ell_m^{[j]}(2i\zeta)^{-m},\quad|\zeta|\to\infty,\;\Im\zeta\ge0,
\]
with real coefficients
$\ell_m^{[j]}=-\int_0^\infty g_m^{[j]}(x)\,dx$,
where
\[
g_1^{[j]}(x)=v_j(x),\quad g_2^{[j]}(x)=-v_j'(x),\quad g_m^{[j]}(x)=-g_{m-1}^{[j]}{}'(x)- \sum_{p=1}^{m-2}g^{[j]}_p (x)g^{[j]}_{m-p-1}(x).
\]
%The first few $\ell_m^{[j]}$'s are
%\begin{multline}\label{eqn:ellm}
%\ell_1^{[j]}=-\int_0^\infty v_j(x)\,dx,\quad \ell_2^{[j]}=-v_j(0),\quad \ell_3^{[j]}=v_j'(0)+\int_0^\infty v_j^2(x)\,dx,\quad\ell_4^{[j]}=-v''(0)+2v^2(0),\\
%\ell_5^{[j]}=v_j'''(0)-6v_j(0)v_j'(0)-\int_0^\infty(v_j'(x)^2+2v_j^3(x))\,dx.\hfill{}
%\end{multline}
%Using this, \eqref{eqn:log-sum} becomes,
%\begin{equation}\label{eqn:log-sum2}
%\log D(\zeta)=\log\left(\frac{K(\zeta)}{in\zeta}\right)+\sum_{m=1}^\infty\left(\ell_m^{[1]}+\cdots+\ell_m^{[n]}\right)(2i\zeta)^{-m}.
%\end{equation}
Also, these $g_m(x)$ occur in the asymptotic expansion
\begin{equation}\label{eqn:g}
\frac{\theta'(x,\zeta)}{\theta(x,\zeta)}=i\zeta+\sum_{m=1}^\infty g_m(x)(2i\zeta)^{-m}.
\end{equation}
Now it remains to find an asymptotic expansion for $\log\left(\frac{K(\zeta)}{in\zeta}\right)$.
%First, we obtain an asymptotic expansion for $\frac{K(\zeta)}{in\zeta}$. From \cite[pg.190]{yafaev},
By \eqref{eqn:g},
\begin{align*}
\frac{K(x,\zeta)}{in\zeta}\equiv\frac{1}{in\zeta}\sum_{j=1}^n\frac{\theta_j'(x,\zeta)}{\theta_j(x,\zeta)} %&=1+\sum_{j=1}^n\sum_{m=1}^\infty \frac{g_m^{[j]}(x)(2i\zeta)^{-m}}{in\zeta}
&=1+\sum_{m=2}^\infty\frac{2}{n}\left(g_{m-1}^{[1]}(x)+\cdots+g_{m-1}^{[n]}(x)\right)(2i\zeta)^{-m}. % [resolved] sum starts at m=2? yes it does Yafaev pg.190 eqn 4.12 for theta'/theta^{-1}-i\zeta
\end{align*}

To take the logarithm, we use the same method as in the half-line case, %in \cite[\S4]{yafaev}, 
which is to find the asymptotic expansion for the logarithmic derivative, and then integrate. First we do long division: % The following lemma comes from long division of $\frac{A'(x,\zeta)}{A(x,\zeta)}$.
\begin{lemma}[logarithmic derivative]
If $A(x,\zeta)=1+\sum_{m=1}^M a_m(x)(2i\zeta)^{-m}+A_M(x,\zeta)$ and $A'(x,\zeta)=\sum_{m=2}^M a_m'(x)(2i\zeta)^{-m}+\tilde{A}_M(x,\zeta)$ are asymptotic expansions, then we have the asymptotic expansion, %also has an asymptotic expansion to order $M$, we have the asymptotic expansion,
\begin{equation}
B(x,\zeta):=\frac{A'(x,\zeta)}{A(x,\zeta)}=\sum_{m=1}^M b_m(x)(2i\zeta)^{-m}+B_M(x,\zeta),
\end{equation}
where
\begin{equation}
b_1(x)=a_1'(x),\quad b_m(x)=a_m'(x)-\sum_{p=1}^{m-1}b_p(x) a_{m-p}(x).
\end{equation}
\end{lemma}

We fix a branch of the logarithm using $\frac{K(\zeta)}{in\zeta}=1+O(|\zeta|^{-1})$ and requiring $\log\left(\frac{K(\zeta)}{in\zeta}\right)\to0$ as $|\zeta|\to\infty$. Integrating, we get the following:
\begin{cor}[logarithm expansion]\label{cor:log-expansion}
Let $A(x,\zeta),B(x,\zeta)$ be as above, and suppose we have $\int_0^\infty b_m(x)\,dx<\infty$ for $m=1,\ldots,M$ and $\int_0^\infty|B_M(x,\zeta)|\,dx\lesssim_{M}|\zeta|^{-M-1}$ for each $\zeta$, $\Im\zeta\ge0$, $|\zeta|\ge c>0$. Then we can integrate term by term to get the asymptotic expansion,
\begin{equation}
\log A(x,\zeta)=\sum_{m=1}^M C_m(x)(2i\zeta)^{-m}+C_M(x,\zeta),\qquad C_m(x)=-\int_x^\infty b_m(y)\,dy.
\end{equation}
\end{cor}
We also define $C_1:=C_1(0)=a_1(0)$, and for $m\ge2$, if $a_m(x)\to0$ as $x\to\infty$,
\begin{align*}
C_m:=C_m(0)&=-\int_0^\infty(a_m'(y)-\sum_{p=1}^{m-1}b_p(y)a_{m-p}(y))\,dy\\
&=a_m(0)+\sum_{p=1}^{m-1}\int_0^\infty b_p(y)a_{m-p}(y)\,dy.
\end{align*}

In our case, $A(x,\zeta)=\frac{K(x,\zeta)}{in\zeta}$, so $a_1(x)\equiv0$, $a_m(x):=\frac{2}{n}\left(g_{m-1}^{[1]}(x)+\cdots+g_{m-1}^{[n]}(x)\right)$, and we get the asymptotic series
\begin{equation}\label{eqn:logA}
\log\left(\frac{K(\zeta)}{in\zeta}\right)=\sum_{m=2}^\infty C_m(2i\zeta)^{-m},\quad|\zeta|\to\infty,\;\Im\zeta\ge0.
\end{equation}
(We show in Appendix~\ref{appendix:remainders} that the functions $A(x,\zeta)$ and $B(x,\zeta)$ in our case do indeed satisfy the necessary hypotheses to apply the corollary.)
The first few $C_m$'s (other than $C_1=0$) are
\begin{multline}\label{eqn:C_m}
C_2=\frac{2}{n}\sum_{j=1}^nv_j(0),\quad C_3=-\frac{2}{n}\sum_{j=1}^nv_j'(0),\quad C_4=\frac{2}{n}\sum_{j=1}^n(v_j''(0)-v_j(0)^2)-\frac{2}{n^2}\left(\sum_{j=1}^n v_j(0)\right)^2,\\
C_5=\frac{2}{n}\sum_{j=1}^n\left(-v_j'''(0)+4v_j'(0)v_j(0)\right)+\frac{4}{n^2}\left(\sum_{j=1}^nv_j(0)\sum_{j=1}^nv_j'(0)\right).\hfill{}
\end{multline}

Equation \eqref{eqn:log-sum} becomes,
\begin{equation}\label{eqn:log-sum-total}
\log D(\zeta)=\sum_{m=1}^\infty \left(C_m+\left(\ell_m^{[1]}+\cdots+\ell_m^{[n]}\right)\right)(2i\zeta)^{-m},
\end{equation}
so $L_m=C_m+\sum_{j=1}^n\ell_{m}^{[j]}$. Now using \eqref{eqn:C_m} and the definition of the $\ell_m$'s, we  get \eqref{eqn:Lm}. %can compute the first few $L_m$'s to get \eqref{eqn:Lm}.

%%%%%%%%%%%%%%%%%%%%%%%%%%%%%%%%%%%%%%%%%%%%%
%%%%%%%%%%%%%%%%%%%%%%%%%%%%%%%%%%%%%%%%%%%%%
\section{Reduction to the real line}\label{sec:real}
%%%%%%%%%%%%%%%%%%%%%%%%%%%%%%%%%%%%%%%%%%%%%
%%%%%%%%%%%%%%%%%%%%%%%%%%%%%%%%%%%%%%%%%%%%%

A star graph with only $n=2$ edges can be identified with the whole real line. So using results for star graphs, we can prove things about scattering on the real line. The real line case with $v$ smooth is handled directly in \cite[\S5]{yafaev}, but we will recover the results for the real line with $v$ smooth away from a point by setting $n=2$ in the star graph case. %In this sense, star graphs are a generalization of the half line and whole line cases.
A potential ${v}$ on $\R$ is viewed as a pair of potentials $v_1(x):={v}(x)$ and $v_2(x):={v}(-x)$, $x\ge0$ on the $n=2$ star graph. Using results for star graphs from \cite{demirel-ssf} along with Theorem~\ref{thm:trace}, we can show:
\begin{cor}[real line]
Consider the Schr\"{o}dinger operator $H=-\frac{d^2}{dx^2}+v$ on the real line and suppose $\int_\R |v(x)|\,dx<\infty$.
Then the following hold.
\begin{enumerate}[(i)]
\item The perturbation determinant is
\begin{align}
D(\zeta)&=-\frac{1}{2i\zeta}\left[\theta_1'(0,\zeta)\theta_2(0,\zeta)+\theta_2'(0,\zeta)\theta_1(0,\zeta)\right]\\
&=-\frac{1}{2i\zeta}W(\theta_2(-\cdot,\zeta),\theta_1(\cdot,\zeta))=:m(\zeta),\label{eqn:real-wronskian}
\end{align}
where $W$ is the Wronskian and $\theta_1,\theta_2$ are the Jost solutions on the two half-line edges.
\item We have the trace formula
\begin{equation}
\operatorname{tr}(R_0(z)-R(z))=\frac{\frac{d}{d\zeta^2}D(\zeta)}{D(\zeta)}=\frac{\dot{m}(\zeta)}{2\zeta m(\zeta)},\quad z=\zeta^2.
\end{equation}
\item (low energy asymptotics) Assume $\int_\R(1+x^2)|v(x)|\,dx<\infty$, and let $W(\zeta):=W(\theta_2(-\cdot,\zeta),\theta_1(\cdot,\zeta))$. If $W(0)=0$, then $\theta_1(x,0)=\alpha\theta_2(-x,0)$ for some $\alpha\in\R\setminus\{0\}$, and in this case,
\begin{equation}
W(\zeta)=-i(\alpha+\alpha^{-1})\zeta+O(|\zeta|^2),\quad|\zeta|\to0.
\end{equation}
\item Assume again $\int_\R(1+x^2)|v(x)|\,dx<\infty$. The Schr\"{o}dinger operator $H$ has a zero energy resonance of order one if $W(0)=0$, and no zero energy resonance if $W(0)\ne0$.
\item (Levinson's formula) Suppose $\int_\R(1+|x|)|v(x)|\,dx<\infty$, and let  $N$ be the number of negative eigenvalues of $H$ and $m\in\{0,1\}$ be the multiplicity of the resonance at $\zeta=0$. If $m=1$, also suppose that $\int_{\R}(1+x^2)|v(x)|\,dx<\infty$. Then
\begin{equation} % [resolved] need smoothness of v at 0? no
\eta(\infty)-\eta(0)=\pi\left(N+\frac{m-1}{2}\right).
\end{equation}
\item (trace formulas) Suppose that $\int_\R(1+|x|)|v(x)|\,dx<\infty$ and \eqref{eqn:potential-smooth} for $j=1,2$ hold. Then we get the trace formulas in Theorem~\ref{thm:trace}. In particular, if $v\in C^\infty(\R)$, the first few $L_m$'s are
\begin{multline}
L_1=-\int_{-\infty}^\infty v(x)\,dx,\quad L_2=0,\quad L_3=\int_{-\infty}^{\infty}v^2(x)\,dx,\quad L_4=0,\quad L_5=-\int_{-\infty}^\infty(v'(x)^2+2v^3(x)).\hfill{}
\end{multline}
\end{enumerate}
\end{cor}

\begin{rmks}\mbox{}
\begin{enumerate}[(i)]
\item In (vi), using Theorem~\ref{thm:trace} allows for a potential $v$ that is \emph{discontinuous} at $x=0$. We emphasize that in this case the trace formulas contain \emph{additional contributions} from the discontinuities of $v$ and its derivatives at $0$.
\item For $v\in C^\infty(\R)$, although we can compute the values $L_m$ using \eqref{eqn:Lm}, a much simpler formula is derived in \cite[\S5.2]{yafaev}. These are
\begin{align}
L_{2m+1}&=-\int_{-\infty}^\infty g_{2m+1}(x)\,dx,\quad L_{2m}=0.
\end{align}
\item The Jost functions $\theta_1,\theta_2$ extend easily from $[0,\infty)$ to $\R$. This follows from the existence proof of Jost solutions on the half-line given in e.g. \cite[\S4.1]{yafaev}, and ensures that \eqref{eqn:real-wronskian} makes sense. The $\theta_1$ here agrees with the Jost solution $\theta_1$ described in \cite[\S5.1]{yafaev}, but because we identify $e_2$ with $[0,\infty)$ rather than $(-\infty,0]$, the $x$ argument in $\theta_2$ must be negated to match the $\theta_2$ in \cite[\S5.1]{yafaev}. %Defining $\tilde{\theta}_2(x,\zeta):=\theta_2(-x,\zeta)$, we get the asymptotics $\tilde{\theta}_2(x,\zeta)=e^{-i\zeta x}(1+o(1))$ as $x\to-\infty$.
\end{enumerate}
\end{rmks}

\begin{proof}
(i) and (ii) follow immediately from Proposition~\ref{prop:D(z)}.
(iii) follows from the proof of Proposition~\ref{prop:low-energy} (low energy asymptotics) found in \cite{demirel-ssf}, though some of the steps are simplified in the case $n=2$. %The proof in this case comes down to considering the case where $\omega_1(0)=\omega_2(0)=0$ and the case where $\omega_1(0),\omega_2(0)\ne0$. In the first case, we end up with $K(0)=0$ and $\dot{K}(0)=i\left[\frac{1}{\theta_1(0,0)^2}+\frac{1}{\theta_2(0,0)^2}\right]$. After Taylor expanding, we get $D(\zeta)=\frac{1}{2}\left(\alpha+\alpha^{-1}\right)+O(|\zeta|)$ as $\zeta\to0$.
%In the second case, since $\theta_1(0,0)=\theta_2(0,0)=0$, Taylor expansion gives us $\theta_j(0,\zeta)=\zeta\dot{\theta_j}(0,0)+O(|\zeta|^2)$, $j=1,2$. Using \cite[Lemma 4.6]{demirel-ssf} to get $\dot{\theta}_j(0,0)=\frac{-i}{\theta_1'(0,0)}$ and substituting into the formula for $D(\zeta)$ gives us the same result as in the first case.
The fact $\alpha\in\R$ comes from $\theta_j(x,0)=\overline{\theta_j(x,0)}$.
(iv) follows from Proposition~\ref{prop:low-energy}. % and continuity of $W(\zeta)$ as $\zeta\to0$. % [no ignore] is W(\zeta) cts as zeta->0? probably yes, we basically defined it to be cts. Or use explicit formula?
(v) follows from the result for star graphs proved in \cite{demirel-ssf}.
For (vi),  because $v\in C^\infty(\R)$, we require $v_1^{(2k)}(x)=v_2^{(2k)}(x)$ and $v_1^{(2k+1)}(x)=-v_2^{(2k+1)}(x)$ for all $k\in\N_0$. Then we just compute $L_m$, $1\le m\le 5$ via \eqref{eqn:Lm}.

\end{proof}

%%%%%%%%%%%%%%%%%%%%%%%%%%%%%%%%%%%%%%%%%%%%%%%%%%%%%%%%%%%%%%%%%%%%%%%%%%%%%%%%%%%%%%%%%%
%%%%%%%%%%%%%%%%%%%%%%%%%%%%%%%%%%%%%%%%%%%%%%%%%%%%%%%%%%%%%%%%%%%%%%%%%%%%%%%%%%%%%%%%%%
\appendix
%%%%%%%%%%%%%%%%%%%%%%%%%%%%%%%%%%%%%%%%%%%%%%%%%%%%%%%%%%%%%%%%%%%%%%%%%%%%%%%%%%%%%%%%%%
%%%%%%%%%%%%%%%%%%%%%%%%%%%%%%%%%%%%%%%%%%%%%%%%%%%%%%%%%%%%%%%%%%%%%%%%%%%%%%%%%%%%%%%%%%

%%%%%%%%%%%%%%%%%%%%%%%%%%%%%%%%%%%%%%%%%%%%%
\section{Remainder estimates for Section~\ref{sec:high-energy}}\label{appendix:remainders}
%%%%%%%%%%%%%%%%%%%%%%%%%%%%%%%%%%%%%%%%%%%%%

Here we show that the asymptotic expansion for $\log\left(\frac{K(\zeta)}{in\zeta}\right)$ derived in Section~\ref{sec:high-energy} is valid, in particular, that Corollary~\ref{cor:log-expansion} applies to our specific $A(x,\zeta)$ and $B(x,\zeta)$. Ignoring $x$, it is easy to verify that $B(x,\zeta)$ has a valid asymptotic series in terms of powers $(2i\zeta)^{-m}$, but to integrate, we need error bounds involving $x$ as well. Recall that
\[
A(x,\zeta)=\frac{K(x,\zeta)}{in\zeta}=1+\sum_{m=2}^M a_m(x)(2i\zeta)^{-m}+A_M(x,\zeta),
\]
where $a_m(x)=\frac{2}{n}\sum_{j=1}^n g_{m-1}^{[j]}(x)$. From properties of the $g_m$'s which can be found in e.g. \cite[\S4.4]{yafaev}, it follows that for all $x\ge0$ and $\Im\zeta\ge0$, $|\zeta|\ge c>0$,
\begin{align}
|A_M(x,\zeta)|&\le C_M|\zeta|^{-M-1}(1+|x|)^{-\rho-M+1}\\
\left|\frac{\partial A_M(x,\zeta)}{\partial x}\right|&\le C_M|\zeta|^{-M-1}(1+|x|)^{-\rho-M}\label{eqn:partial}\\
a_m^{(j)}(x)&=O(x^{-1-(m-1)(\rho-1)-j}),\quad x\to\infty.\label{eqn:a-m}
\end{align}

\begin{lemma}\label{lem:A} % checked this lemma 9/27/15
There is the asymptotic expansion $A(x,\zeta)^{-1}=1+\sum_{m=2}^M\alpha_m(x)(2i\zeta)^{-m}+\alpha_M(x,\zeta)$, with estimates for all $x\ge0$ and $\Im\zeta\ge0$, $|\zeta|\ge c>0$,
\begin{align}
|\alpha_M(x,\zeta)|&\le\tilde{C}_M|\zeta|^{-M-1}(1+x)^{-1-M(\rho-1)}\label{eqn:error}\\
|\alpha_m(x)|&\le\tilde{c}_M2^m(1+x)^{-1-(m-1)(\rho-1)},\quad 2\le m\le M.\label{eqn:alpha-m}
\end{align}
\end{lemma}
\begin{proof}
Using the power series expansion for $(1+(A(x,\zeta)-1))^{-1}$ (since $A(x,\zeta)-1=O(|\zeta|^{-2})$), we get
\begin{equation}\label{eqn:inverse-expansion}
A(x,\zeta)^{-1}=1+\sum_{k=1}^\infty(-1)^k\left(\sum_{m=2}^Ma_m(x)(2i\zeta)^{-m}+A_M(x,\zeta)\right)^k.
\end{equation}
By \eqref{eqn:a-m} with $j=0$, we get for $x\ge x_0>0$, 
\begin{equation}\label{eqn:a-m2}
|a_m(x)|\le c_1x^{-1-(m-1)(\rho-1)}\le c_2(1+x)^{-1-(m-1)(\rho-1)}, \quad 2\le m\le M.
\end{equation}
With a possibly larger value of $c_2$ the right-most bound for $|a_m(x)|$ holds for all $x\ge0$. %(These constants depend on $M$.) %even though we do not explicitly indicate this. 
From expanding \eqref{eqn:inverse-expansion}, the coefficient $\alpha_m(x)$ on $(2i\zeta)^{-m}$, $2\le m\le M$, is 
\[
\alpha_m(x)=\sum_{k=1}^{\lfloor m/2\rfloor}\sum_{\substack{p_1+\cdots+p_k=m\\ (p_i)\in\{2,3,\ldots\}^k}}a_{p_1}(x)\cdots a_{p_k}(x).
\]
(Note any term with a nonzero power of $A_M(x,\zeta)$ will be at least $O(|\zeta|^{-M-1})$.) %so will not contribute to any $\alpha_m(x)$.) 
Summing over all $k$, there are $2^{m-1}$ possible sequences\footnote{This is just the number of compositions of $m$. Also, since there are no powers of $(2i\zeta)^{-1}$, we actually have far fewer elements to sum over, but the $2^{m-1}$ bound will be sufficient.} %\footnote{If we have $m$ dots and $m-1$ spaces between the dots, a composition $(p_i)$ can be identified with choosing whether or not to place a divider for each space.} 
$(p_i)\in\N^k$ satisfying $\sum p_i=m$.
Each product $a_{p_1}(x)\cdots a_{p_k}(x)$ is $O(x^{-1-(m-1)(\rho-1)})$ regardless of $k$ since $\rho-1\le1$. %; this is because $\rho-1\le1$ so that for example, $a_i(x)a_j(x)=O(x^{-2-(m-2)(\rho-1)})=O(x^{-1-(m-1)(\rho-1)})$. % checked: k+(m-k)(rho-1) >= 1 + (m-1)(rho-1)
Thus for each $2\le m\le M$ and $x\ge0$, there is some $\tilde{c}_M$ so that
\[
|\alpha_m(x)|\le\tilde{c}_M2^m(1+x)^{-1-(m-1)(\rho-1)},
\]
as claimed.

It remains to show the error estimate \eqref{eqn:error}. %Since we will estimate all terms by their norm, we are free to rearrange the remaining terms. 
We have two types of error terms: terms with a nonzero power of $A_M(x,\zeta)$, and the remaining higher order terms with $(2i\zeta)^{-M-j}$, $j\ge1$. We deal with the higher order terms quickly; note that the coefficient on $(2i\zeta)^{-M-j}$ can be estimated in a similar way as $\alpha_m(x)$. By \eqref{eqn:a-m2} and since there are only $M-1$ different $a_m(x)$'s and at most $\lfloor (M+j)/2\rfloor\le M+j$ of the $a_m(x)$ terms in each product, we get the error is bounded by %(wlog assuming $c_2\ge1$)
\begin{align*}
\sum_{j=1}^\infty 2^{M+j}\max(1,c_2^{M+j})(1+x)^{-1-(M+j-1)(\rho-1)}|\zeta|^{-M-j}& %(2c_2|\zeta|^{-1})^M(1+x)^{-1-(M-1)(\rho-1)}\sum_{j=1}^\infty(2c_2(1+x)^{-(\rho-1)}|\zeta|^{-1})^j\\
\le c_M|\zeta|^{-M-1}(1+x)^{-1-M(\rho-1)},
\end{align*}
which agrees with \eqref{eqn:error}. 

%From expanding the $k$th powers in \eqref{eqn:inverse-expansion}, we first have
For the terms with powers of $A_M(x,\zeta)$, we first have
\[
|A_M(x,\zeta)|+|A_M(x,\zeta)|^2+\cdots\lesssim C_M|\zeta|^{-M-1}(1+|x|)^{-\rho-M+1},\quad|\zeta|\to\infty.
\]
The other terms are mixed with powers $(2i\zeta)^{-m}$ and their coefficients. For a given $m$ and composition $(p_i)$ summing to $m$, we have the mixed terms
\[
a_{p_1}(x)\cdots a_{p_k}(x)(2i\zeta)^{-m}\left[\pm\binom{k+1}{k}A_M(x,\zeta)\mp\binom{k+2}{k}A_M(x,\zeta)^2\pm\binom{k+3}{k}A_M(x,\zeta)^3\mp\cdots\right].
\]
The maximum possible binomial coefficient on any $A_M(x,\zeta)^j$ is 
\[
\binom{\lfloor m/2\rfloor+j}{\lfloor m/2\rfloor}=\binom{\lfloor m/2\rfloor+j}{j}\le\binom{m+j}{j}\le\frac{(m+j)^j}{j!}\le \frac{e^j}{\sqrt{2\pi j}}(1+\frac{m}{j})^j\le e^je^m.
\]
Summing over all compositions of $m$ and then over all $m$, we get the error bound
\begin{align*}
\sum_{m=2}^\infty\sum_{j=1}^\infty2^{m-1}&|2i\zeta|^{-m}c_2^m(1+x)^{-1-(m-1)(\rho-1)}|A_M(x,\zeta)|^je^je^m\\
&\lesssim C_M|\zeta|^{-M-1}(1+x)^{-\rho-M+1}(1+x)^{-2+\rho}\sum_{m=2}^\infty (c_2e|\zeta|^{-1}(1+x)^{-(\rho-1)})^m,\quad|\zeta|\to\infty\\
&\lesssim C_M|\zeta|^{-M-3}(1+x)^{-2\rho-M+1}\le C_M|\zeta|^{-M-1}(1+x)^{-1-M(\rho-1)},\quad|\zeta|\to\infty,
\end{align*}
which implies \eqref{eqn:error}.

%$\binom{\lfloor m/2\rfloor+j}{\lfloor m/2\rfloor}=\binom{\lfloor m/2+j\rfloor}{j}$. Summing over all compositions of $m$ and then over all $m$, we have the error bound
%\begin{equation}\label{eqn:remainder2}
%\sum_{m=2}^\infty\sum_{j=1}^\infty2^{m-1}|A_M(x,\zeta)|^j|2i\zeta|^{-m}c_2^m(1+x)^{-1-(m-1)(\rho-1)}\binom{\lfloor m/2\rfloor+j}{j}.
%\end{equation}
%Since
%\[
%\binom{\lfloor m/2\rfloor+j}{j}\le\binom{m+j}{j}\le\frac{(m+j)^j}{j!}\le \frac{e^j}{\sqrt{2\pi j}}(1+\frac{m}{j})^j\le e^je^m,
%\]
%\eqref{eqn:remainder2} is bounded by
%\begin{align*}
%\sum_{j=2}^\infty 2^{-m}&|\zeta|^{-m}2^{m-1}c_2^m(1+x)^{-1-(m-1)(\rho-1)}e^m\sum_{j=1}^\infty |A_M(x,\zeta)|^j e^j\\
%&\lesssim C_M|\zeta|^{-M-1}(1+x)^{-\rho-M+1}(1+x)^{-2+\rho}\sum_{m=2}^\infty (c_2e|\zeta|^{-1}(1+x)^{-(\rho-1)})^m,\quad|\zeta|\to\infty\\
%&\lesssim C_M|\zeta|^{-M-3}(1+x)^{-2\rho-M+1}\le C_M|\zeta|^{-M-1}(1+x)^{-1-M(\rho-1)},\quad|\zeta|\to\infty,
%\end{align*}
%which implies \eqref{eqn:error}.

\end{proof}

\begin{lemma}\label{lem:B}
There is the asymptotic expansion $B(x,\zeta)=A'(x,\zeta)A(x,\zeta)^{-1}=\sum_{m=1}^Mb_m(x)(2i\zeta)^{-m}+B_M(x,\zeta)$, with remainder estimate for all $x\ge0$ and $\Im\zeta\ge0$, $|\zeta|\ge d>0$,
\begin{equation}
|B_M(x,\zeta)|\le c_M|\zeta|^{-M-1}(1+x)^{-\rho-M(\rho-1)}.
\end{equation}
\end{lemma}
As a result,
\begin{align*}
|C_M(x,\zeta)|&=\left|\int_x^\infty B(y,\zeta)-\sum_{m=1}^M\int_x^\infty b_m(y)\,dy\,(2i\zeta)^{-m}\right|\le\int_x^\infty|B_M(y,\zeta)|\,dy\lesssim_M|\zeta|^{-M-1},
\end{align*}
so the expansion for $\log A(x,\zeta)$ obtained in \eqref{eqn:logA} is indeed a valid asymptotic expansion.

\begin{proof}[Proof (of Lemma~\ref{lem:B})]
The asymptotic expansion of order $M$ for $A'(x,\zeta)$ is just
\[
A'(x,\zeta)=\sum_{m=2}^M a_m'(x)(2i\zeta)^{-m}+\frac{\partial A_M(x,\zeta)}{\partial x}.
\]
Multiply the asymptotic expansions of order $M$ for $A'(x,\zeta)$ and $A(x,\zeta)^{-1}$, then use \eqref{eqn:partial}, \eqref{eqn:a-m} with $j=1$ (adapted for $x\ge0$ like in \eqref{eqn:a-m2}), and Lemma~\ref{lem:A} equations \eqref{eqn:error} and \eqref{eqn:alpha-m} to deal with the error terms.
\end{proof}


\begin{thebibliography}{10}
\bibitem{book} G. Berkolaiko and P. Kuchment, \textit{Introduction to Quantum Graphs}, Mathematical Surveys and Monographs Vol. 186 (2010).

\bibitem{BuslaevFaddeev} V. S. Buslaev and L. D. Faddeev, \textit{Formulas for traces for a singular Sturm-Liouville differential operator}, Soviet Math. Dokl. 1 (1960), 451-454.

\bibitem{demirel-ssf} S. Demirel, \textit{The spectral shift function and Levinson's theorem for quantum star graphs}, Journal of Mathematical Physics Vol. 53 Issue 8 (2012).

\bibitem{app-ssf} F. Gesztesy and K. A. Makarov, \textit{Some applications of the spectral shift operator}, Operator theory and its applications 25, 267 (2000). %arXiv:\href{http://arxiv.org/pdf/math/9903186v1.pdf}{math/9903186}.

\bibitem{GeMiZi} F. Gesztesy, M. Mitrea and M. Zinchenko, \textit{Variations on a theme of Jost and Pais}, J. Funct. Anal. 253, Issue 2, 399--448 (2007).

\bibitem{JoPa} R. Jost and A. Pais, \textit{On the scattering of a particle by a static potential}, Phys. Rev. 82, 840--851 (1951).

\bibitem{yafaev} D. R. Yafaev, \textit{Mathematical Scattering Theory: Analytic Theory}, Mathematical Surveys and Monographs Vol. 158 (2010).

\bibitem{yafaev1} D. R. Yafaev, \textit{Mathematical Scattering Theory: General Theory}, Translations of Mathematical Monographs Vol. 105 (1992).


\end{thebibliography}
\end{document}